\documentclass[conference]{IEEEtran}
 \usepackage[table,xcdraw]{xcolor}
\usepackage{graphicx}
\usepackage[export]{adjustbox}
\usepackage{cite}
\usepackage{amsmath,amssymb,amsfonts}
\usepackage{amssymb}
\usepackage{mathtools}
\usepackage{textcomp}
\usepackage{algpseudocode}
\usepackage{algorithm,algcompatible,amsmath}
\usepackage{graphicx, caption, subcaption}
\usepackage{mwe}
\usepackage[utf8]{inputenc}
\usepackage[english]{babel}
\usepackage{amsthm}
\newcommand{\vect}[1]{\boldsymbol{#1}}
\algnewcommand\INPUT{\item[\textbf{Input:}]}%
\algnewcommand\OUTPUT{\item[\textbf{Output:}]}%
\usepackage{tikz}
\usetikzlibrary{positioning} 
\usetikzlibrary{calc}
\usepackage{pgfplots}
\pgfplotsset{compat=newest}
\pgfplotsset{plot coordinates/math parser=false}
\newlength\figureheight
\newlength\figurewidth 
\newcommand{\RNum}[1]{\uppercase\expandafter{\romannumeral #1\relax}}

\def\BibTeX{{\rm B\kern-.05em{\sc i\kern-.025em b}\kern-.08em
		T\kern-.1667em\lower.7ex\hbox{E}\kern-.125emX}}
	\usepackage{filecontents}
\usepackage{array}
\usepackage{booktabs}
\usepackage{graphics}
\usepackage{makeidx}
\usepackage{nomencl}
\usepackage{url}
\usepackage{notoccite}
\usepackage{subcaption}
\newtheorem{mydef}{Definition}

\DeclareMathOperator{\Tr}{Tr}
\newtheorem{lem}{Lemma}
\newtheorem{thm}{Theorem}

\theoremstyle{definition}
\newtheorem{exmp}{Example}[section]

\ifCLASSINFOpdf
\else
\fi

\begin{document}
%
\title{Impact of Spatial Correlation in MIMO Radar}

\author{\IEEEauthorblockN{Aya Mostafa Ahmed and
		Aydin Sezgin}
	\IEEEauthorblockA{Institute of Digital Communication Systems, \\Faculty of Electrical and Computer Engineering,\\ Ruhr-Universität Bochum, Germany\\
		Email: \{aya.mostafaibrahimahmad; aydin.sezgin\}@rub.de}
\and
\IEEEauthorblockN{Eduard A. Jorswieck}
\IEEEauthorblockA{Institute for Communications Technology,\\ Faculty of Electrical and Computer Engineering,\\ TU Braunschweig, Germany\\
Email: e.jorswieck@tu-bs.de}}


%


\maketitle

\begin{abstract}
	The impact of spatial correlation on mutual information (MI) is analyzed for MIMO radar. Unlike the work done in literature for statistical MIMO radar, we consider the spatial correlation of the target matrix elements to study the correlated MIMO radar performance. There is a trade-off between coherent processing gain in correlated  MIMO radar and spatial diversity gain of target scatterers in uncorrelated MIMO radar. We address how the MI between  the received signal and  target channel matrix is affected by spatial correlation. Using majorization theory and the notion of Schur-convexity, we prove that  MI has a changing behavior with respect to spatial correlation, where at low $\mathsf{SNR}$, the MI is Schur-convex, i.e. showing increasing performance as  correlation increases. However, this behavior changes at high $\mathsf{SNR}$, since MI is Schur-concave at high $\mathsf{SNR}$, hence it decreases as the spatial correlation increases. Moreover, we investigate the conditions for spatially uncorrelated MIMO radar. According to these conditions, as the operating frequency increases with respect to the target location and dimensions, the received paths become more uncorrelated. Hence, the setup with lower operating frequency (more correlated) performs better compared to the higher frequency setup at low SNR. However at high $\mathsf{SNR}$, this behavior is reversed.
\end{abstract}
%
\IEEEpeerreviewmaketitle
\section{Introduction}

%

It has been recently shown that applying multiple input multiple output (MIMO) concept in radar systems leads to significant performance improvement \cite{mimoradaridea}. Unlike phased array radar, MIMO radar offers waveform diversity capabilities, sending different  transmit signals, that can be correlated or uncorrelated, and jointly processing the received signals at the receive antennas. Fundamentally, MIMO radar offers more degrees of freedom and more resolution than the phased array radar \cite{MIMODOF}. Such radars can be classified into  colocated or widely distributed (statistical) radar. In colocated MIMO radar, the transmitter and receiver are relatively close, such that the radar observes the same target's radar cross section (RCS). In this case, the radar offers better resolution, higher parameter identifiability and higher sensitivity to detect slow targets \cite{colocatedmimo}. A MIMO radar with widely separated antennas is called statistical MIMO radar. This type of radar captures the spatial diversity of the target's RCS, and with the aid of non-coherent processing,  diversity gain for target detection and parameter estimation can be obtained \cite{MIMOWS}. Moreover, the authors in \cite{MIMOWS} show that by utilizing spatial diversity in statistical MIMO radars, it can overcome bandwidth limitations and offer high resolution target localization. In addition, they derive conditions for spatial de-correlation of the reflected paths to achieve the diversity gain. Those conditions are influenced by the antenna spacing, operating frequency and the target location and dimensions.\\
For both types of radar, the corresponding waveform design problem has been under an on-going research, to optimize target detection or information. In \cite{blum}, the authors proposed waveform design  for MIMO radar to maximize the conditional mutual information (MI) between the target random impulse response and the reflected waveforms. It is shown that waveforms that maximize the MI, also minimize the minimum mean square error (MMSE). The authors in \cite{infoWVD} done similar work but in the presence of colored noise, they show that the optimum waveform in this case should match the target and noise eigen directions.\\
In this paper, we investigate waveform design to maximize MI for statistical MIMO radar. We vary the spatial correlation in different SNR conditions, and analyze how the MI is affected. We use majorization theory and the notion of Schur-convexity to describe analytically this behavior. Specifically, we modify the operating frequency to change the degree of spatial correlation at the receiver. This is due to the fact that, at low operating frequencies, the reflected paths from the target become more correlated, and the other way around for higher frequencies.
Interestingly, we show that MI behaves better under low $\mathsf{SNR}$ condition for correlated channels, however at high $\mathsf{SNR}$, less correlated channels achieve better behavior.\\
The rest of the paper is organized as follows: Section \ref{maj} provides a brief and comprehensive introduction to majorization, and other definitions related to Schur-convexity. Section \ref{SystemModel} presents the statistical MIMO radar model, and analyzes the conditions for spatial de-correlation of a MIMO radar channel. Section \ref{wvd} presents a measure of spatial correlation, and discusses the optimum waveform design for MIMO radar. This section examines the Schur-convexity of the MI function in high and low SNR, then numerical results is provided in section \ref{sim}. Section \ref{conc} draws conclusions.

\section{Preliminaries and Basic Definitions}
\label{maj}
Spatial correlation among the signals received at the receiver has great impact on the performance of the MIMO radar. Since, highly correlated signals would possibly increase the coherent processing gain \cite{Mine}, while uncorrelated signals would emphasize more the spatial diversity of the target scatterers \cite{BLUMspatial}. Therefore, we need to analyze and model the spatial correlation observed from the different paths between the transmitter and receiver. In \cite{MIMOWS}, the authors provide conditions for correlated and uncorrelated MIMO radars, however in this paper we analyze how the spatial correlation affects the system performance. \\Before proceeding with the description of correlation in our model, we introduce some necessary definitions in the following.
\begin{mydef}
we say $\mathbf{x}$ majorizes $\mathbf{y}$ with notion $\mathbf{x}$ $\succeq$ $\mathbf{y}$ if \cite{Jorswieck2007f}
	\begin{equation*}
	\sum_{k=1}^{m} x_k \geq \sum_{k=1}^{m} y_k, m=1,\hdots, n-1 \hspace{0.1in} \text{and}  \sum_{k=1}^{n} x_k = \sum_{k=1}^{n} y_k
	\end{equation*} 
\end{mydef}
Majorization describes a partial order between two vectors $\mathbf{x}$ , $\mathbf{y}$ $\in$ $\mathbb{R}^n$\cite{marshall11}, it depicts if the components of $\mathbf{x}$ is less spread out or more nearly equal than the components of $\mathbf{y}$. The next definition describes the behavior of function $f$ when applied to vectors $\mathbf{x}$ and $\mathbf{y}$. 
\begin{mydef}
	A function $f$ defined on $\mathcal{A}$ $\subset$ $\mathbb{R}^n$ is said to be Schur-convex on $A$ if 
	\begin{equation*}
	\mathbf{x} \succeq \mathbf{y} \hspace{0.05in} \text{on} \hspace{0.05in}  \mathcal{A} \implies f(\mathbf{x}) \geq f(\mathbf{y}),
	\end{equation*}
	and Schur-concave on $\mathcal{A}$ if 
	\begin{equation*}
	\mathbf{x} \succeq \mathbf{y} \hspace{0.05in} \text{on} \hspace{0.05in}  \mathcal{A} \implies f(\mathbf{x}) \leq f(\mathbf{y}).
	\end{equation*}
\end{mydef}
The next lemma provides a condition to test the Schur convexity of a valued vector function.
\begin{lem}[Schur-Ostrowski Condition, {\cite[Lemma 2.5]{Jorswieck2007f}}]
	\label{schurcond}	Let $\mathcal{I}$ $\subset$ $\mathbb{R}$ be an open interval and let $f$ : $\mathcal{I}^n$$\rightarrow$ $\mathbb{R}$ be continously differentiable. $f$ is said to be Schur-convex on $\mathcal{I}^n$ if \\ 
	\begin{center}
		$f$ is symmetric \footnote{A function is symmetric if the argument vector can be arbitraly permuted without changing the value of the function} on $\mathcal{I}$,
	\end{center}
and for all $\mathbf{a}$  $\in$ $\mathcal{I}^n$ 
	\begin{equation}
	\label{schurcond2}
	(x_i-x_j)(\frac{\partial f}{\partial x_i}- \frac{\partial f}{\partial x_j})\geq 0 \hspace{0.1in} \forall \hspace{0.05in}  1 \leq i, j \leq n,
	\end{equation}
	and Schur-concave if the inequality in \eqref{schurcond2} is in the opposite direction \ \cite{Jorswieck2007f}.
\end{lem}
The symmetry condition in Lemma \ref{schurcond} limits its applicability to only symmetric functions. Hence, there have been several works to deal with this restriction. Hwang in \cite{partialordered} generalized the  Schur condition in Lemma \ref{schurcond} for partially ordered sets. He introduced a corresponding notion for the Schur-Ostrowski condition, where  $\frac{\partial f}{\partial x_i} 
\geq \frac{\partial f}{\partial x_j}$ for all $x$ $\in$ $\mathbb{R}^n$ and $i,j= 1,\hdots, n$ where $j$ dominates $i$ in the partially order points and the resulting inequalities ($j > i$).
\begin{thm}\cite{partialordered}
	\label{majorizationgenr}
	Let $f(x_1,\hdots,x_n)$ be a function defined over the domain $\mathcal{D}$, such that $\mathbf{x} = \left[x_1,x_2,\ldots,x_n \right]^T$. Let $\vect{P} = \left[p_1,p_2,\ldots,p_n \right]^T$  be a set of points partially ordered by '$\geq$', and $\vect{a} = \left[a_1,a_2,\ldots,a_n \right]^T$, $\vect{b} = \left[b_1,b_2,\ldots,b_n \right]^T$ be two set of weights where $a_i$ and $b_i$ are associated with $p_i$ for $i=1,\hdots,n$, then \\
	\begin{center}
		$f(a_1,\hdots,a_n)$ $\geq$ $f(b_1,\hdots,b_n)$,
	\end{center}
for all $\vect{a}$ majorizing  $\vect{b}$ on  $\vect{P}$ if and only if $f$ for every $i$ and $j$, $p_i$ $\geq$ $p_j$ fulfills
\\
\begin{equation*}
	 \frac{\partial f}{\partial x_i}\geq \frac{\partial f}{\partial x_j} \hspace{0.2in} \forall \hspace{0.1in}\mathbf{x} \in \mathcal{D}.
\end{equation*}
\end{thm}
	The following definition provides a measure for correlation to compare between two covariance matrices.
\begin{mydef}[{\cite[Definition 4.2]{Jorswieck2007f}}]
	\label{def1}
	If we have two arbitrary target covariance matrices, $\mathrm{\mathbf{R}}_{\bar{\mathbf{h}}}^1$ and $\mathrm{\mathbf{R}}_{\bar{\mathbf{h}}}^2$,  with eigenvalues $\boldsymbol{\sigma}_{h_{1}}$ , and $\boldsymbol{\sigma}_{h_{2}}$ respectively, arranged in descending order such that $\sigma_{h_{1},1}\geq \sigma_{h_{1},2}\geq \hdots \geq\sigma_{h_{1},T}\geq 0$ and $\sigma_{h_{2},1}\geq \sigma_{h_{2},2}\geq \hdots \geq\sigma_{h_{2},T}\geq 0$, where $T=\mathrm{MN}$, with constraint that $\Tr(\mathrm{\mathbf{R}}_{\bar{\mathbf{h}}}^1)$=$\Tr(\mathrm{\mathbf{R}}_{\bar{\mathbf{h}}}^2)$. we say that $\mathrm{\mathbf{R}}_{\bar{\mathbf{h}}}^1$ is more correlated than $\mathrm{\mathbf{R}}_{\bar{\mathbf{h}}}^2$, if $\boldsymbol{\sigma}_{h_{1}}$ $\succcurlyeq$ $\boldsymbol{\sigma}_{h_{2}}$ such that
	\begin{equation}
	\sum_{l=1}^{\mathrm{L}}\sigma_{h_{1},l} \geq \sum_{l=1}^{\mathrm{L}}\sigma_{h_{2},l} \hspace{0.1in}  \text{for} \hspace{0.1in} 1\leq L \leq T-1.
	\end{equation} 
\end{mydef}
This definition is different from the usual statistical correlation definition. Normally in statistics, a diagonal covariance matrix is uncorrelated, independent from the values of auto-covariances on its diagonal. In definition \eqref{def1}, the target covariance matrices are uncorrelated, if the auto-covariances on the diagonal are equal in addition to the statistical independence \cite[Remark 4.1]{Jorswieck2007f}.
\\
This means that the larger the sum of the first $l$ eigenvalues of the covariance matrix of the target are the more correlated are the scattepaths arriving at the receiver from the target. This leads to further insight that if the covariance matrix of a radar target having the most uncorrelated paths, would have equal eigen values, while the target covariance matrix with the most correlated paths would have only one non-zero eigen value.

\section{System Model}
\label{SystemModel}
Assume a distributed target consisting of $Q$ scatterers, each scatterer is considered as independent, and isotropic. The target is illuminated by statistical MIMO radar with widely separated antennas with $M$ transmitters and $N$ receivers as in Figure \ref{fig:targets}, with transmitter $m$ at position $t_m$$=$ $(x_{tm}$, $y_{tm})$, and receiver $n$ at position  $r_n$$=$ ($x_{rn}$, $y_{rn}$). The scattered signal from one scatterer $q$ located at position $t_q=(x_q,y_q)$, received at $r_n$ at time instant $k$ is given by
\begin{equation}
y_n^{q}(k)=\sum_{m=1}^{M}h^{q}_{mn}s_m(k-\tau_{tm}(t_q)+\tau_{rn}(t_q))+w_n(k),
\end{equation}
where $s_m(k)$ is the waveform transmitted by transmitter $m$, $w_n(k)$ is the noise at receiver $n$. Defining $h_{mn}$ as the channel from $m$ to
 receiver $n$ for all $Q$ scatterers, which is given by 
\begin{equation}
h_{mn}=\sum_{q=1}^{Q}\alpha_q \exp(-j2\pi f_c[\tau_{tm}(t_q)+\tau_{rn}(t_q)])
\end{equation}
defining $\alpha_q$ as the reflectivity of the scatterer, which is a zero mean, i.i.d complex Gaussian random variable with variance of $1/Q$ \cite{MimoRadar}, and  $\tau_{tm}(t_q)$$=$$d(t_m,t_q)/c$ is the propagation time delay between transmitter $m$ located at position $t_m$ and scatterer $q$, where $d(t_m,t_q)$ is the distance between  $m$ and $q$, 	and $c$ is the speed of light. Accordingly $\exp(-j2\pi f_c\tau_{tm}(t_q))$ is the phase shift due to the propagation from $m$ to $q$, and similarly $\exp(-j2\pi f_c\tau_{rn}(t_q))$ is the phase shift due to propagation from scatterer $q$ till receiver $n$, where $\tau_{rn}(t_q)$ is the propagation time delay between $q$ and $n$. Similar to \cite{MIMOWS}, we assume that the bandwidth of the waveform transmitted is not wide enough to resolve individual scatterers. Therefore, we assume that $s_m(k-\tau_{tm}(t_q)+\tau_{rn}(t_q))$ $\approx$ $s_m(k-\tau_{tm}(t_0)+\tau_{rn}(t_0))$, where we assume that the radar cross section of the target (RCS) has center of gravity located at $t_0=(x_0,y_0)$.
Furthermore, the path gains $h_{mn}$ is organized in a $N \times M$ matrix $\mathbf{H}$, as shown in \cite{MIMOWS}, the structure of this matrix is
\begin{equation}
\mathbf{H}=\mathbf{K}\mathbf{\Sigma}\mathbf{G}.
\end{equation} 
The transmit paths are organized in a $Q\times M$ matrix $\mathbf{G}$, where $\mathbf{G}=[\mathbf{g}_1^T;\hdots;\mathbf{g}_Q^T]$, where $\mathbf{g}_q^T=[\exp(-j2\pi f_c\tau_{t1}(t_q)),\hdots,\exp(-j2\pi f_c\tau_{tM}(t_q))]$. The receive paths are in a $N\times Q$ matrix $\mathbf{K}$, where  $\mathbf{K}$$=$ $[\mathbf{k}_1,\hdots,\mathbf{k}_Q]$, and $\mathbf{k}_q^T=[\exp(-j2\pi f_c\tau_{r1}(t_q)),\hdots,\exp(-j2\pi f_c\tau_{rN}(t_q))]$. The reflectivity of all scatterers is organized in a diagonal $Q\times Q$ matrix $\boldsymbol{\Sigma}$, where 
$\boldsymbol{\Sigma}=$$\text{diag}([\alpha_1,\hdots,\alpha_Q])$.
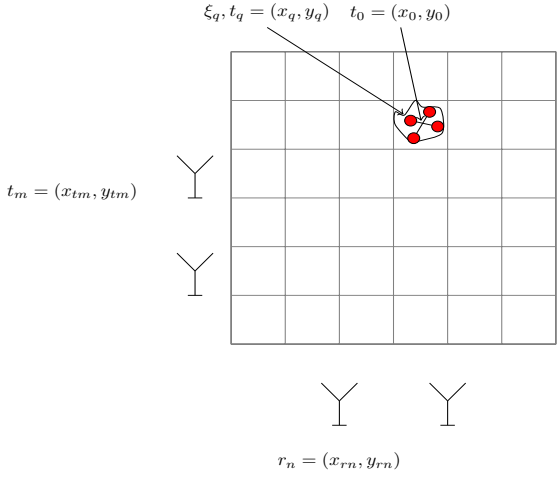
\begin{figure}
	\begin{frame}

	\resizebox{3in}{2.5in}{


		\begin{tikzpicture}[scale=0.5,every node/.style={scale=0.5}] 

		\draw [step=1.5cm,gray] (0,0) grid  (9,9) rectangle (0,0);

		\draw [black] plot [smooth cycle] coordinates {(5,6.2) (4.55,6.75) (4.55,7.15) (4.8,7.05) (5.1,7.5) (5.3,7.3)(5.8,7.25)(5.85,6.5)(5.2,6.25)};

		\node [draw=black, fill=red, circle, radius=0.2cm] (a) at (5.06,6.34) {} ;

		\node [draw=black, fill=red, circle, radius=0.2cm] (b) at (5.725,6.7) {} ;

		\node [draw=black, fill=red, circle, radius=0.2cm] (c) at (5.5,7.15) {} ;

		\node [draw=black, fill=red, circle, radius=0.2cm] (d) at (4.975,6.88) {} ;

		\draw [black] (a) to (c);

		\draw [black] (b) to (d);


		\draw [black] (-1.2,4.5) to (-.8,4.5);

		\draw [black] (-1,4.5) to (-1,5.25) node[below left =.05cm and .7cm,scale=1.5] {$t_m=(x_{tm},y_{tm})$}; 

		\draw [black] (-1,5.25) to (-.5,5.8);

		\draw [black] (-1,5.25) to (-1.5,5.8);


		\draw [black] (-1.2,1.5) to (-.8,1.5);

		\draw [black] (-1,1.5) to (-1,2.25); 

		\draw [black] (-1,2.25) to (-.5,2.8);

		\draw [black] (-1,2.25) to (-1.5,2.8);


		\draw [black] (2.8,-2.5) to (3.2,-2.5);

		\draw [black] (3,-2.5) to (3,-1.75) node[below = .7cm,scale=1.5] {$r_n=(x_{rn},y_{rn})$}; 

		\draw [black] (3,-1.75) to (3.5,-1.2);

		\draw [black] (3,-1.75) to (2.5,-1.2);


		\draw [black] (5.8,-2.5) to (6.2,-2.5);

		\draw [black] (6,-2.5) to (6,-1.75); 

		\draw [black] (6,-1.75) to (6.5,-1.2);

		\draw [black] (6,-1.75) to (5.5,-1.2);

		\draw [->](1,9.75) node [above,scale=1.5] {$\xi_q,t_q=(x_q,y_q)$} to ($(d)+(-0.2,0.2)$);

		\draw [->](4.7,9.75) node [above,scale=1.5] {$t_0=(x_0,y_0)$} to ($(d)+(0.29,0.0)$);

		\end{tikzpicture}


	}

\end{frame}
	\caption{MIMO radar with an extended target, containing of four point targets.}
	\label{fig:targets}
\end{figure}
\\
Therefore, we can obtain the total received signal across all $K$ time samples as
\begin{equation}
\label{perrx}
\mathbf{y}_n=\mathbf{h}_n^T\mathbf{S}^T+\mathbf{w}_n,
\end{equation}
 where $\mathbf{y}_n=[y_{n}(1)\:y_{n}(2) \hdots y_{n}(K)]$, $\mathbf{h}_n = [h_{1n}\:h_{2_n} \hdots h_{Mn}]^T$, $\mathbf{S}=[\mathbf{s}(1) \: \mathbf{s}(2) \hdots \mathbf{s}(K)]^T$, where $\mathbf{s}(k)=[s_1(k)\: s_2(k) \hdots s_M(k)]$. We assume that $K \geq \max(M,N)$.
From \eqref{perrx}, we define the received signal from all the antennas as
\begin{equation}
\label{withoutvec}
\mathbf{Y}=\mathbf{S}\mathbf{H}+\mathbf{W},
\end{equation}
 in which $\mathbf{Y}$ $\in$ $\mathbb{C^\mathrm{K\times N}}$, $\mathbf{Y}=[\mathbf{y}_1 \: \mathbf{y}_2 \hdots \mathbf{y}_N]$, $\mathbf{H}$ $\in$ $\mathbb{C^\mathrm{M\times N}}$ is the target scattering matrix containing all the path gains $h_{mn}$ from transmit to receive antennas, $\mathbf{W}$ $\in$ $\mathbb{C^\mathrm{K\times N}}$ is a colored noise matrix with independent and identically (i.i.d) distributed columns, where $ \mathbf{W}=[\mathbf{w}_1 \: \mathbf{w}_2 \hdots \mathbf{w}_N]$. Moreover we define $\bar{\mathbf{y}}=\text{vec}(\mathbf{Y})$, $\bar{\mathbf{h}}=\text{vec}(\mathbf{H})$, and $\bar{\mathbf{w}}=\text{vec}(\mathbf{W})$, where $\text{vec}(\mathbf{X})$ is obtained by column wise staking of the matrix $\mathbf{X}$.  Consequently, \eqref{withoutvec} can be rewritten as
\begin{equation}
\bar{\mathbf{y}}=\tilde{\mathbf{S}}\bar{\mathbf{h}}+\bar{\mathbf{w}},
\end{equation}
 where $\tilde{\mathbf{S}}=\mathbf{I}_{N}\otimes\mathbf{S}$. We assume that $\mathbf{H}$ and $\mathbf{W}$ are independent, with distributions \[
\bar{\mathbf{h}} \sim \mathcal{CN}(0,\mathrm{\mathbf{R}}_{\bar{\mathbf{h}}})\,,
\] 
 \[
\bar{\mathbf{w}} \sim \mathcal{CN}(0,\mathrm{\mathbf{R}}_{\bar{\mathbf{w}}})\,,
\] 
where $\mathrm{\mathbf{R}}_{\bar{\mathbf{h}}}$ $\in$ $\mathbb{C^\mathrm{MN\times MN}}$ is positive semidefinite correlation matrix of the target, defined as 
$\mathrm{\mathbf{R}}_{\bar{\mathbf{h}}}$$=\mathop{\mathbb{E}}[\bar{\mathbf{h}} \bar{\mathbf{h}}^\mathrm{H}]$ and  $\mathrm{\mathbf{R}}_{\bar{\mathbf{w}}}$ $\in$ $\mathbb{C^\mathrm{NK\times NK}}$ is a positive semidefinite correlation matrix of the noise. Let the eigen-decomposition of $\mathrm{\mathbf{R}}_{\bar{\mathbf{h}}}$ and $\mathrm{\mathbf{R}}_{\bar{\mathbf{w}}}$ be
\begin{equation*}
\mathrm{\mathbf{R}}_{\bar{\mathbf{h}}}=\mathrm{\mathbf{V}_h}\boldsymbol{\Sigma}_h\mathrm{\mathbf{V}_h^H},
\end{equation*}
\begin{equation*}
\mathrm{\mathbf{R}}_{\bar{\mathbf{w}}}=\mathrm{\mathbf{V}_w}\boldsymbol{\Sigma}_w\mathrm{\mathbf{V}_w^H},
\end{equation*}
where $\mathrm{\mathbf{V}_h}$, $\mathrm{\mathbf{V}_w}$ are unitary matrices, while $\boldsymbol{\Sigma}_h$, and $\boldsymbol{\Sigma}_w$ are diagonal matrices, with vectors $\boldsymbol{\sigma}_h$, $\boldsymbol{\sigma}_w$ on the diagonals respectively, such that  $\boldsymbol{\sigma}_h=([\sigma_{h,1},\sigma_{h,2},\hdots,\sigma_{h,\mathrm{MN}})]$, $(\boldsymbol{\sigma}_w)=([\sigma_{w,1},\sigma_{w,2},\hdots,\sigma_{w,\mathrm{NK}}])$ are diagonal matrices whose elements are arranged in descending order.

Suppose that there are two transmit antennas at location $(x_{tm},y_{tm})$ and $(x_{ti},y_{ti})$ respectively, while the receive ones are $(x_{rn},y_{rn})$ and $(x_{rj},y_{rn})$ respectively. Furthermore, the target dimensions is defined as $d_x$ along x axis and $d_y$ along y axis.  If at least one of the following conditions is met, then the channel is considered as uncorrelated.
\\
	There are four conditions for spatial de-correlation of the channel elements $h_{mn}$ \cite{MIMOWS}.
	\begin{equation}
	\label{eq:conditions}
	\begin{split}
	\frac{x_{tm}}{d(t_m,t_0)}-\frac{x_{ti}}{d(t_m,t_0)} > \frac{\lambda_c}{d_x}\\
	\frac{y_{t_m}}{d(t_m,t_0)}-\frac{y_{ti}}{d(t_m,t_0)} > \frac{\lambda_c}{d_y}\\
	\frac{x_{rn}}{d(r_n,t_0)}-\frac{x_{rj}}{d(r_n,t_0)} > \frac{\lambda_c}{d_x}\\
	\frac{y_{rn}}{d(r_n,t_0)}-\frac{y_{rj}}{d(r_n,t_0)} > \frac{\lambda_c}{d_y},
	\end{split}
	\end{equation}

where $\lambda_c$ is the operating wavelength. As noticed from the previous conditions, changing any of the following factors would affect the spatial de-/correlation of the channel matrix,
\begin{enumerate}
	\label{factors}
	\item Spacing between transmit / receive antennas
	\item Operating frequency 
	\item Target Dimensions
	\item Distance between the target and the antennas.
\end{enumerate}
Consequently, those factors would affect the eigenvalue distribution of the target covariance matrix, which would in turn affect the Schur-convexity/Schur-concavity of the MI. For further insights into those conditions, let us apply what was previously discussed in section \ref{sectionMIschur}, where on one hand having a spatially correlated channel matrix $\mathbf{H}$ is better at low $\mathsf{SNR}$ from MI perspective, while on the other hand a de-correlated channel is better at high $\mathsf{SNR}$.
\section{Optimum waveform design and impact of spatial correlation}
\label{wvd}
The measure of correlation defined in \eqref{def1} allows us to analyze the impact of spatial correlation on performance measures for waveform design. Indeed, we will investigate how the waveform design for maximizing the mutual information (MI) between $\bar{\mathbf{y}}$ and $\bar{\mathbf{h}}$ can be affected by the spatial correlation of $\bar{\mathbf{h}}$.
\subsection{Waveform Design based on maximizing Mutual Information}
The mutual information between $\bar{\mathbf{y}}$, and $\bar{\mathbf{h}}$, if the transmitted waveform is known, is given by \cite{infoWVD}
\begin{equation}
I(\bar{\mathbf{y}};\bar{\mathbf{h}}|\tilde{\mathbf{S}}) = N[\log[\det(\tilde{\mathbf{S}} \mathrm{\mathbf{R}}_{\bar{\mathbf{h}}}\tilde{\mathbf{S}}^H+\mathrm{\mathbf{R}}_{\bar{\mathbf{w}}}]-\log\det(\mathrm{\mathbf{R}}_{\bar{\mathbf{w}}})]] .
\end{equation}
Then, the optimization problem of waveform design to maximize the MI can be formulated as 
\begin{equation}
\label{main}
\begin{aligned}
& \underset{\tilde{\mathbf{S}}}{\text{max}}
& & \log[\det(\tilde{\mathbf{S}} \mathrm{\mathbf{R}}_{\bar{\mathbf{h}}}\tilde{\mathbf{S}}^H\mathrm{\mathbf{R}}_{\bar{\mathbf{w}}}^{-1}+\mathrm{\mathbf{I}_{NK}})] \\
& \text{s.t.}
& &  \Tr{(\tilde{\mathbf{S}}\tilde{\mathbf{S}}^H)} \leq P_\mathsf{tot}.
\end{aligned}
\end{equation}
\begin{lem}\cite{Fiedler}
	\label{MIoptimlemma}
The  optimum waveform for maximizing MI is the following
\begin{equation}
\label{optim}
\mathbf{\tilde{S}_{opt}}=\mathrm{\mathbf{V}_w} \left[\mathbf{0}_{\mathrm{MN}\times(\mathrm{NK}-\mathrm{MN})} \hspace{0.2in}\mathrm{\mathbf{\Sigma_s}}^{1/2}\right]^T\mathrm{\mathbf{V}_H}^H.
\end{equation}  
 $\mathrm{\mathbf{\Sigma_s}}$ is a square diagonal matrix, $\mathrm{\mathbf{\Sigma_s}}$ $\in$ $\mathbb{C^\mathrm{MN\times MN}}$ with elements $\sigma_{s,i}$ on its diagonal.
\end{lem}
It should be mentioned that in \eqref{optim}, the left singular vector of the  optimum waveform refers to the eigenvector of the noise covariance matrix in increasing order, while the right singular values refer to the eigen vector of the covariance matrix which should be in decreasing order, i.e. the eigenvalues of the noise and the target are sorted in oppositional order according to the following theorem.
\begin{thm}\cite{Fiedler}
	\label{fiedler}
For positive semidefinite matrices $\mathrm{\mathbf{A}}$ and  $\mathrm{\mathbf{B}}$, with eigenvalues $\alpha_1 \ge \alpha_2...\ge \alpha_n$, $\beta_1 \ge \beta_2...\ge \beta_n$.
\begin{equation}
\prod_{i=1}^{n}(\alpha_i+\beta_i)\leq \det(\mathbf{A} + \mathbf{B})\leq\prod_{i=1}^{n}(\alpha_i+\beta_{n+1-i}).
\end{equation}
Hence, if the eigen value decomposition of $\mathrm{\mathbf{A}}=\mathrm{\mathbf{U}_A}\boldsymbol{\Lambda}_A\mathrm{\mathbf{U}_A^H}$ and $\mathrm{\mathbf{B}}=\mathrm{\mathbf{U}_B}\boldsymbol{\Lambda}_B\mathrm{\mathbf{U}_B^H}$, then the upper bound is achieved for $\mathrm{\mathbf{U}_A}=$$\mathbf{P}\mathrm{\mathbf{U}_B}$, where $\mathbf{P}$ is a permutation matrix with ones on the anti-diagonal such that
\begin{center}
$ \mathbf{P}=
\begin{bmatrix}
0 & 0 &\hdots & 1 \\
0 &\hdots & 1& 0\\
\vdots & \vdots & \vdots &\vdots \\
1 & 0 & \hdots & 0
\end{bmatrix}$,
\end{center}
\end{thm}
and the lower bound is achieved for $\mathrm{\mathbf{U}_A}=$$\mathrm{\mathbf{U}_B}$.
Then we can solve for the  power allocation of the singular values $\sigma_{s,i}$ of the optimal waveform $\mathbf{\tilde{S}_{opt}}$ in \eqref{optim} by rewriting \eqref{main} as                                       
\begin{equation}
\label{maineig}
\begin{aligned}
& \underset{\sigma_{s,i}}{\text{max}}
& & \sum_{i=1}^{\mathrm{MN}}\log\left(\frac{\sigma_{s,i}\hspace{0.03in}\sigma_{h,i}}{\sigma_{w,\mathrm{MN}-i+1}}+1 \right) \\
& \text{s.t.}
& &  \sum_{k=1}^{M}\sigma_{s,k} \leq P_\mathsf{tot}.
\end{aligned}
\end{equation}
Then we can obtain the solution using the celebrated water filling algorithm \cite{infoWVD}, such that 
\begin{equation*}
\sigma_{s,i}=\left(\frac{1}{\lambda} - \frac{\sigma_{w,\mathrm{MN}-i+1}}{\sigma_{h,i}} \right)^+,
\end{equation*}
where $\lambda$ is the waterlevel and is determined based on the total power, by solving the following equation
\begin{equation*}
\sum_{i=1}^{\mathrm{MN}}\left(\frac{1}{\lambda} - \frac{\sigma_{w,\mathrm{MN}-i+1}}{\sigma_{h,k}} \right)^+=P_\mathsf{tot}.
\end{equation*}
\subsection{Analysis of effect of spatial correlation on MI}
\label{sectionMIschur}
In this subsection, we analyze the MI expression, if it is Schur-convex or Schur-concave with respect to the eigenvalues of the target covariance matrix, and subsequently how the function behaves with respect to the correlation of signals reflected from the target scatterers. As per Lemma \ref{MIoptimlemma}, the eigenvalues of the noise and the target are assumed to be in oppositional order to obtain the optimum solution, as explained in Theorem \ref{fiedler}.  Therefore, we rewrite \eqref{maineig} as
\begin{equation}
\label{func}
f\left(\sigma_{h,k}\right)=\sum_{i=1}^{\mathrm{MN}}\log\left(\frac{\sigma_{s,i}\hspace{0.03in}\sigma_{h,i}}{\sigma_{w,\mathrm{MN}-i+1}}+1 \right).
\end{equation} 
Hence, to use Theorem \ref{majorizationgenr}, we assume that $\vect{\sigma_h}$ is a partially ordered vector,  $\sigma_{h,i}$ $>$ $\sigma_{h,j}$. Therefore, we can use Theorem \ref{majorizationgenr} to check for the Schur condition with respect to the eigenvalue of $\boldsymbol{\sigma}_h$ by taking the partial derivative of \eqref{func} such that
\begin{equation}
\frac{\partial f}{\partial \sigma_{h,i}}=\frac{\sigma_{s,i}}{\sigma_{h,i}\sigma_{s,i}+\sigma_{w,\mathrm{MN}-i+1}}.
\end{equation}
Since elements of $\boldsymbol{\sigma}_h$ are arranged in descending order, $\left(\sigma_{h,i} - \sigma_{h,j}\right) \geq 0$. Hence, the sign of
\begin{equation*}
\frac{\partial f}{\partial \sigma_{h,i}}-\frac{\partial f}{\partial \sigma_{h,j}},
\end{equation*}
which is defined as
\begin{equation}
\label{schurapplied}
\frac{\sigma_{s,i}}{\sigma_{h,i}\sigma_{s,i}+\sigma_{w,\mathrm{MN}-i+1}}-\frac{\sigma_{s,j}}{\sigma_{h,j}\sigma_{s,j}+\sigma_{w,\mathrm{MN}-j+1}},
\end{equation}
is totally dependent on the optimum power allocation values and the noise eigenvalues.
Herein, the behavior of the function will be analyzed at high and low $\mathsf{SNR}$.
\begin{lem}
	\label{highlowsnrnoncolored}
In case of non-colored, independent, identically distributed (i.i.d) noise, in high $\mathsf{SNR}$ regimes, the  water-filling solution to \eqref{maineig} is given by  $\boldsymbol{\sigma}_s=\frac{P_\mathsf{tot}}{\mathrm{MN}}\boldsymbol{1}^T$ (equal power allocation $p$), hence, \eqref{schurapplied} would be always smaller than zero, hence Schur-concave. 
However in low $\mathsf{SNR}$ regimes, the solution of \eqref{maineig} would be $\boldsymbol{\sigma}_s=\left[P_{\mathsf{tot}},0,\hdots,0\right]$, where the power is only given for the strongest eigen mode of the target. Consequently, \eqref{schurapplied} would be always positive, since the second term in \eqref{schurapplied} would be $0$, and the first term is positive, then according to Lemma \ref{schurcond}, the function is Schur-convex.
\end{lem}

\begin{thm}
	\label{schurcolored}
In case of colored-noise, in high $\mathsf{SNR}$ regimes, \eqref{schurapplied} is Schur-convex if 
\begin{equation}
\label{schurcondcolr}
\max_{1 \leq i < j \leq MN}  \frac{\sigma_{h,i}-\sigma_{h,j}}{\sigma_{w,MN-j+1}-\sigma_{w,MN-i+1}} \leq \frac{1}{p},
\end{equation}
and Schur-concave otherwise.	
\end{thm}
\begin{proof}
We can further simplify \eqref{schurapplied} to be the following 
\begin{equation}
\label{schursimplify}
\left(\sigma_{h,i}+\frac{\sigma_{w,MN-i+1}}{\sigma_{s,i}}\right)^{-1}-\left(\sigma_{h,j}+\frac{\sigma_{w,MN-j+1}}{\sigma_{s,2}}\right)^{-1}.
\end{equation}
Hence, in order for \eqref{schursimplify} to be greater than 0, then the following must apply
\begin{equation}
\label{schursimplify2}
\sigma_{h,i}+\frac{\sigma_{w,MN-i+1}}{\sigma_{s,i}} \leq \sigma_{h,j}+\frac{\sigma_{w,MN-j+1}}{\sigma_{s,j}},
\end{equation}
since in high $\mathsf{SNR}$ regimes, the optimal water-filling solution is nearly equal power allocation $\sigma_{s,i}$$=$$\sigma_{s,j}$$=$$p_l$. Therefore, after some mathematical reordering in \eqref{schursimplify2}, we can get the result in \eqref{schurcondcolr}.
	\end{proof}
\begin{exmp}
	\label{example}
If we assumed $MN=4$, $\boldsymbol{\sigma}_h=[5,2,1,0.5]$ and $\boldsymbol{\sigma}_w=[8,4,3,2]$. Then we have 6 cases demonstrated in Table \ref{table} with their corresponding values of the left hand side (L.H.S) of \eqref{schurcondcolr}. The maximum value of \eqref{schurcondcolr} here occurs when $i=1$, and $j=2$, therefore in order to apply Theorem \ref{schurcolored}, then $p$ $\in (0,\:\frac{1}{3}]$.
\begin{table}[]
	\centering
	\begin{tabular}{@{}|c|c|c|c|@{}}
		\toprule
		instance & i & j & (18)  \\ \midrule
		1        & 1 & 2 & 3   \\ \midrule
		2        & 1 & 3 & 2     \\ \midrule
		3        & 1 & 4 & 0.75 \\ \midrule
		4        & 2 & 3 & 1   \\ \midrule
		5        & 2 & 4 & 0.3   \\ \midrule
		6        & 3 & 4 & 0.125 \\ \bottomrule
	\end{tabular}
		\caption{Evaluation of \eqref{schurcondcolr} using values in example \ref{example}}
		\label{table}
\end{table}
\end{exmp}
In low $\mathsf{SNR}$, the effect of colored noise will not be significant, as only the first eigen-mode of the target would be triggered, hence lemma \ref{highlowsnrnoncolored} will hold as well in case of low $\mathsf{SNR}$ with colored noise.\\
This changing behavior of the MI in low and high $\mathsf{SNR}$, gives indication that according to Definition 3 and 4, spatially correlated channels behave better in low $\mathsf{SNR}$, however in high $\mathsf{SNR}$, it is better to have uncorrelated channel.
%
\\
In the following section, we simulate this changing behavior by controlling the spatial correlation. Moreover, for further insights, we manipulate the spatial correlation conditions for MIMO radar which was previously discussed in section \ref{SystemModel}, by changing the operating frequency and analyze its effect in low and high $\mathsf{SNR}$.
\section{Simulations}
\label{sim}
\subsection{Schur Convexity and Schur Concavity of MI }
In the first set of simulations, the performance of the MI function is analyzed across different spatial correlations. In theorem \ref{highlowsnrnoncolored}, it is proven that MI has changing behavior in high and low $\mathsf{SNR}$ regimes. Here, we illustrate this behavior through numerical evaluation.  We assume that $M=N=2$ and $K=2$. The eigenvalues of $\mathbf{R}_w$ for colored noise case are $[8,4,3,2]$. We keep the eigenvalues of the noise fixed, and change the total power value to vary the $\mathsf{SNR}$. In order to simulate the effect of correlation, the eigenvalues of $\mathbf{R}_h$ are defined as $\boldsymbol{\sigma_h}$$=$$\tau*[1,0,0,0]+(1-\tau)[0.25,0.25,0.25,0.25]$, hence, the eigenvalues  will vary from uncorrelated when $\tau=0$ to highly correlated when $\tau=1$. 
 \begin{figure}
		\centering
%
%
\definecolor{mycolor1}{rgb}{0.00000,0.44700,0.74100}%
\definecolor{mycolor2}{rgb}{0.85000,0.32500,0.09800}%
\definecolor{mycolor3}{rgb}{0.92900,0.69400,0.12500}%
\begin{tikzpicture}

\begin{axis}[%
width=3in,
height=2in,
at={(0.758in,0.481in)},
scale only axis,
xmin=0,
xmax=1,
xlabel style={font=\color{white!15!black}},
xlabel={$\tau$},
ymin=0.4,
ymax=1,
ylabel style={font=\color{white!15!black}},
ylabel={Normalized Mutual Information},
axis background/.style={fill=white},
xmajorgrids,
ymajorgrids,
legend style={at={(0.395,0.265)}, anchor=south west, legend cell align=left, align=left, draw=white!15!black}
]
\addplot [color=mycolor1, line width=1.5pt, mark=+, mark options={solid, mycolor1}]
  table[row sep=crcr]{%
0	0.567787075935456\\
0.1	0.622142366654015\\
0.2	0.674343432882135\\
0.3	0.723267576476507\\
0.4	0.771267380777282\\
0.5	0.81844099970817\\
0.6	0.861173429745502\\
0.7	0.900153648332987\\
0.8	0.935987736203103\\
0.9	0.969145837719839\\
1	1\\
};
\addlegendentry{SNR = 0dB}

\addplot [color=mycolor2, line width=1.5pt, mark=o, mark options={solid, mycolor2}]
  table[row sep=crcr]{%
0	0.859058794929678\\
0.1	0.887186117564985\\
0.2	0.90548178054544\\
0.3	0.916474733304863\\
0.4	0.922276553127326\\
0.5	0.925074486325087\\
0.6	0.927907565763826\\
0.7	0.935342352774133\\
0.8	0.953062123516244\\
0.9	0.977446537206346\\
1	1\\
};
\addlegendentry{SNR = 5dB}

\addplot [color=mycolor3, line width=1.5pt, mark=triangle, mark options={solid, mycolor3}]
  table[row sep=crcr]{%
0	1\\
0.1	0.996872085989904\\
0.2	0.987361100104969\\
0.3	0.972188556116229\\
0.4	0.951221748014807\\
0.5	0.923596398972602\\
0.6	0.887489737129871\\
0.7	0.839378606044602\\
0.8	0.77192100685014\\
0.9	0.66656155635414\\
1	0.497406765219843\\
};
\addlegendentry{SNR = 20dB}

\end{axis}
\end{tikzpicture}%
		\caption{Normalized Mutual Information (MI) (with respect to the maximum value) as function of $\tau$ which represents the degree of correlation ($\tau = 0$ totally uncorrelated, $\tau$ = 1 totally correlated channel) for different total $\mathsf{SNR}$ values (0 dB,5 dB, and 20 dB) assuming colored noise.}
		\label{fig:total}
	\end{figure}
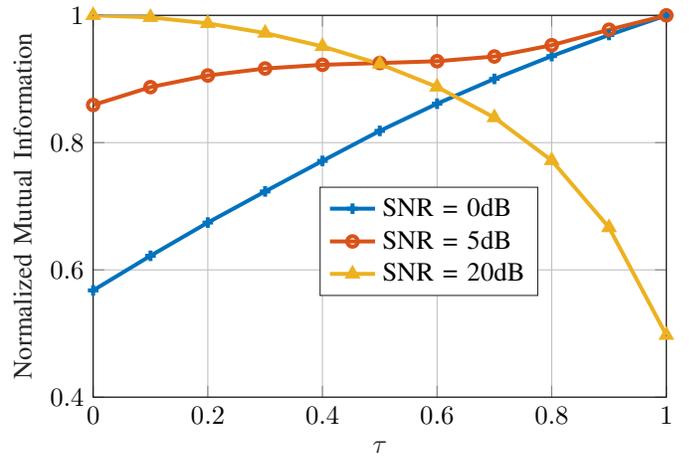
%
\\ In figure \ref{fig:total}, the MI is plotted across different $\mathsf{SNR}$ values, the MI is normalized at each $\mathsf{SNR}$, where at  $\mathsf{SNR}$ 0, it can be seen that the MI is increasing as the correlation increases. Since the MI is Schur-convex at low $\mathsf{SNR}$, it increases with increasing the correlation. However, when the $\mathsf{SNR}$ is increased to be 20 dB, the function has a decreasing behavior since it is Schur-concave at high $\mathsf{SNR}$, where the maximum of MI is achieved when $\tau$$=$$0$, and then decreases with increasing the correlation. Yet, at intermediate $\mathsf{SNR}$ at 5 dB, the function is not behaving neither Schur- convex nor concave.
\subsection{Spatially correlated MIMO Radar setup}
In the second set of simulations, we simulate the scenario in Figure \ref{fig:targets} using the model of a widely separated MIMO radar with $M$$=2$ and $N$$=2$. Here, we want to check the effect of the operating frequency on the spatial correlation conditions in \eqref{eq:conditions}. Hence, we carefully chose the other factors defined in eq. \eqref{eq:conditions} such that it will not affect the correlation, to verify the effect of frequency. Hence, the coordinates of the transmitter are (2,4.8) and (2.2,4) meters, while the receivers are located at (0,2) and (0,4). We assume that there is a distributed target with $\mathrm{Q}=1000$, its center is located at (2,2) meters, and $\mathrm{d}_x=\mathrm{d}_y=2\mathrm{m}$. The parameters are chosen such that we have two different channels $\mathbf{H}$, one spatially correlated by violating the rules in \eqref{eq:conditions}, and the other matrix spatially de-correlated. From the four factors stated, we changed the first working frequency to be $f_c=0.1$ GHz, which makes the channel correlated, while in the second case $f_c=8$ GHz, decreasing $\lambda_c$ and therefore obeying the mentioned conditions. Figure \ref{fig:MI} shows the performance of MI at both frequencies, which agrees with the behavior explained before, since the low frequency curve performs better at low $\mathsf{SNR}$, where the spatial correlation is high.  This agrees with corollary \ref{schurcolored}  where at low $\mathsf{SNR}$ the function is Schur-convex. However, as the $\mathsf{SNR}$ increases, the high frequency curve achieves higher MI, since the spatial correlation decreases. Accordingly, to achieve maximum MI at high $\mathsf{SNR}$, the channel elements must be de-correlated.
\begin{figure}[!h]
%
%
\begin{tikzpicture}

\begin{axis}[%
width=3in,
height=1.5in,
at={(0.758in,0.481in)},
scale only axis,
xmin=-10,
xmax=10,
xlabel style={font=\color{white!15!black}},
xlabel={SNR (dB)},
ymin=0,
ymax=20,
ylabel style={font=\color{white!15!black}},
ylabel={Mutual Information},
axis background/.style={fill=white},
xmajorgrids,
ymajorgrids,
legend style={legend cell align=left, align=left,anchor=east, draw=white!15!black}
]
\addplot [color=blue, line width=1.0pt]
  table[row sep=crcr]{%
-10	2.39678250527724\\
-9	2.7802239387949\\
-8	3.11863408853907\\
-7	3.53763654999734\\
-6	3.94801380933002\\
-5	4.46795330654987\\
-4	5.06572799515862\\
-3	5.54154786393133\\
-2	6.15293926590859\\
-1	6.79924412024768\\
0	7.59614978002305\\
1	8.38102143741275\\
2	8.9816785515242\\
3	9.80600819502294\\
4	10.8290882693529\\
5	11.742736086759\\
6	12.6196822738764\\
7	13.5375625458619\\
8	14.4272746514607\\
9	15.3717660475215\\
10	16.3830352434754\\
};
\addlegendentry{0.1 GHz}

\addplot [color=red, line width=1.0pt,mark=asterisk]
table[row sep=crcr]{%
-10	1.70084595092424\\
-9	1.94248463072349\\
-8	2.35886623187378\\
-7	2.72570633091118\\
-6	3.25692963746924\\
-5	3.7546032953036\\
-4	4.17778977102001\\
-3	4.89505363208945\\
-2	5.5764258858695\\
-1	6.40979694197789\\
0	7.20616258849724\\
1	8.04148280816706\\
2	8.98642878500295\\
3	10.024541682644\\
4	11.1284455913656\\
5	12.3543043632498\\
6	13.4685981587349\\
7	14.7464480557701\\
8	15.7091641862801\\
9	17.0356216057105\\
10	18.2988515254482\\
};
\addlegendentry{8 GHz}

\end{axis}

\end{tikzpicture}%
		\caption{Mutual Information as function of $\mathsf{SNR}$ for two different operating frequencies at 0.1 GHz (highly correlated channel) and 8 GHz (less correlated channel) showing the change in behavior of MI function in high and low $\mathsf{SNR}$ .}
		\label{fig:MI}
	\end{figure}
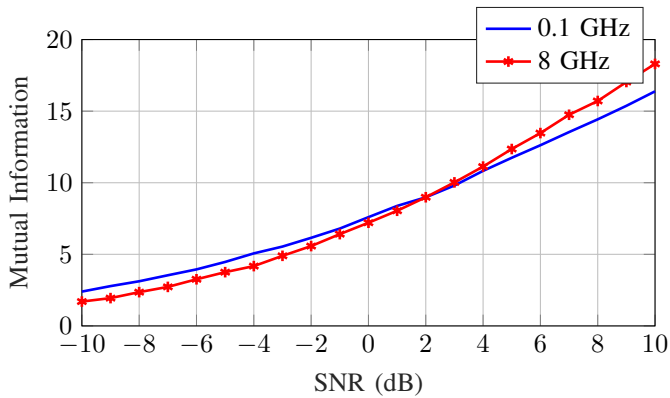

\section{Conclusion}
\label{conc}
In this paper, we discussed the effect of spatial correlation in a statistical MIMO radar. We used MI between the target random response and the reflected signal as a metric in presence of colored noise. We proved that MI is a Schur-convex function with respect to spatial correlation at low SNR, i.e monotonically increasing function. Contrarily, this behavior changes at high SNR, and the function is Schur-concave. Moreover, we applied those findings on statistical MIMO radar setup, by changing the operation frequency to control the spatial correlation of the reflected paths. The simulations show that at low SNR, the performance of the radar is better at low frequencies, which is surpassed by the high frequency operating radar at high SNR conditions.


\section*{Acknowledgment}

This work was supported by the German Research Foundation (DFG) for the CRC/TRR 196 [MARIE] under project S0-3.



%
%
%
\bibliographystyle{IEEEtran}
\bibliography{conf}

\end{document}